\documentclass[preprint,12pt]{elsarticle}

\usepackage{amsmath,amssymb,amsthm,mathtools}
\usepackage{setspace}
\usepackage{hyperref}
\usepackage{enumitem}
\usepackage{booktabs}
\usepackage{xcolor}
\usepackage{graphicx}
\usepackage{float}
\usepackage{threeparttable}
\usepackage{tabularx}
\usepackage{array}
\usepackage[ruled,linesnumbered]{algorithm2e}

\graphicspath{{figures/}}

\newcommand{\thead}[1]{\multicolumn{1}{c}{#1}}

\newcolumntype{R}{>{\raggedleft\arraybackslash}X}
\newcolumntype{L}{>{\raggedright\arraybackslash}X}

\SetAlFnt{\small}
\SetAlCapFnt{\normalsize}
\SetAlCapNameFnt{\normalsize}

\hypersetup{%
  colorlinks=true,%
  linkcolor=blue, citecolor=blue, urlcolor=blue,%
  pdftitle={When to Trust Confidence Thresholding: Calibration Diagnostics for Pseudo-Labelled Regression}%
}

\newtheorem{theorem}{Theorem}
\newtheorem{proposition}[theorem]{Proposition}

\theoremstyle{definition}

\theoremstyle{remark}
\newtheorem{remark}{Remark}

\newcommand{\E}{\mathbb{E}}

\newcommand{\Var}{\operatorname{Var}}

\newcommand{\R}{\mathbb{R}}
\newcommand{\N}{\mathcal{N}}
\newcommand{\1}{\mathbf{1}}

\newcommand{\dto}{\xrightarrow{\;d\;}}
\newcommand{\plim}{\operatorname{plim}}
\newcommand{\sgn}{\operatorname{sgn}}
\newcommand{\Vs}{V^{*}}

\journal{arXiv}

\begin{document}
\begin{frontmatter}
\title{When to Trust Confidence Thresholding: Calibration Diagnostics for Pseudo-Labelled Regression}

\author[ucl]{Marcell T.\ Kurbucz\corref{cor1}}

\ead{m.kurbucz@ucl.ac.uk}

\cortext[cor1]{Corresponding author.}

\affiliation[ucl]{organization={Institute for Global Prosperity, The Bartlett, University College London},%
             addressline={9--11 Endsleigh Gardens},%
             city={London},%
             postcode={WC1H 0EH},%
             country={United Kingdom}}

\begin{abstract}
\small
Calibrated probability outputs of trained classifiers are increasingly used as inputs to downstream regression estimands such as effects, prevalences, or disparities for a latent group observed only on a small labelled subset. A standard practice is to threshold the calibrated score at a confidence cutoff and treat the hard label as the truth. Building on a recent identification result for the underlying moment equation, we develop a calibration-aware diagnostic apparatus for pseudo-labelling pipelines. We derive a closed-form expression for the attenuation bias that confidence thresholding induces in the downstream regression coefficient, and show that the bias can be predicted, before any inference is run, from the residual score variance $\Vs=\E[\Var(p\mid X)]$ on the unlabelled set after partialling out the downstream controls~$X$. We further obtain a sharp sensitivity bound under bounded calibration drift, and identify the boundary $\Vs=0$, which holds iff $p$ is a deterministic function of $X$; this motivates a structural separation between classifier features~$W$ and downstream controls $X\subsetneq W$. Five controlled simulations and a UCI Adult illustration trace the predictions. The contribution is operational: a $(\Vs, \kappa)$ decision rule that practitioners can compute from any classifier output to decide whether confidence thresholding is safe.
\end{abstract}

\end{frontmatter}

\section{Introduction}
\label{sec:intro}

\noindent
Modern machine-learning pipelines increasingly emit calibrated probability scores as primary outputs rather than intermediate classification decisions, and these scores are then used as inputs to a downstream inferential task: a clinical risk classifier feeds a population-level resource-allocation calculation; a fairness audit uses a proxy classifier for an unobserved protected attribute and estimates a disparity in some outcome \cite{kallus2022assessing}; prevalence estimation in biomedical imaging labels a large unlabelled population and regresses the outcome of interest on the predicted class; pattern recognition systems for medical, biometric, and behavioural data feed calibrated risk scores into operational decision rules. In each case the classifier is trained on a small labelled subset $L$ and applied to a much larger unlabelled subset $U$, with the resulting probability used in a downstream regression.

The dominant practice in pseudo-labelling \cite{lee2013pseudo, sohn2020fixmatch, zhang2021flexmatch, wang2023freematch, tarvainen2017mean} converts the calibrated probability to a hard label by thresholding at some confidence level $\tau_{thr}$, typically $0.9$ or $0.95$. The choice has been recognised as empirically consequential \cite{sohn2020fixmatch, zhang2021flexmatch, arazo2020pseudo}, but no quantitative account exists of how thresholding biases a downstream estimand.

This paper supplies that quantitative account in a deliberately restricted regression setting that admits a sharp closed-form analysis. The downstream task is the partial association between a latent binary group indicator $G$ and an outcome $Y$, controlling for a covariate vector $X$, formalised in Section~\ref{sec:setup-data} as the partially linear regression \eqref{eq:plr}; the classifier is trained on a (possibly richer) feature set $W \supseteq X$ to produce a calibrated probability $p = f(W)$ on the unlabelled set $U$.

The companion identification paper \cite{kurbucz2026joe} establishes that the partial association $\tau$ in \eqref{eq:plr} is point-identified by a moment equation in the calibrated probability $p$, with $\Vs=\E[\Var(p\mid X)]$ as the asymptotic-precision quantity. That paper is identification-theoretic; it does not analyse the pseudo-labelling practice that motivated the question, nor does it construct a practitioner-facing decision rule. The present paper takes the identification result as given and develops its consequences for pseudo-labelling.

Concretely, we make three contributions, formalised in Section~\ref{sec:theory} as five predictions (P1--P5) and verified empirically in Section~\ref{sec:experiments}: (i) a closed-form attenuation factor $\kappa(\tau_{thr})$ for hard pseudo-labelling, expressed as a ratio of residual covariance moments and computable on the unlabelled set alone---the central new object of this paper; (ii) a bias-variance regret result for confidence-thresholded abstention (FixMatch-style), with closed-form bias and an explicit variance inflation, neither of which appears in \cite{kurbucz2026joe}; and (iii) a one-page diagnostic procedure, the $(\widehat{V^{*}}, \widehat\kappa)$ decision rule, that operationalises the choice between the supervised, soft, and hard estimators (Algorithm~\ref{alg:diag}). The sandwich-variance result and the calibration-drift sensitivity bound, recalled in Section~\ref{sec:setup-anchor} and Section~\ref{sec:p4} for self-containedness, are taken from the companion paper and applied here in the pseudo-labelling setting.

The paper does not propose a new pseudo-labelling algorithm, predict classification accuracy, or compete with FixMatch on its native loss: it studies a single-pass regression on pseudo-labels and provides a calibration-aware diagnostic for the downstream estimand. Five controlled simulations and a public-data illustration on UCI Adult \cite{kohavi1996scaling} trace the predictions. The remainder is organised as follows. Section~\ref{sec:related} situates the work in the literature; Section~\ref{sec:setup} formalises the setup and recalls the anchor identification result; Section~\ref{sec:theory} develops the five predictions and the diagnostic procedure; Section~\ref{sec:experiments} reports the experiments; Section~\ref{sec:discussion} discusses operational implications and limitations.

\section{Related work}
\label{sec:related}

\noindent
The work draws on four threads. Pseudo-labelling \cite{lee2013pseudo} converts predicted class probabilities of unlabelled examples into hard labels and trains the classifier on the augmented set. FixMatch \cite{sohn2020fixmatch} adds a consistency loss and restricts the pseudo-label to confident examples; FlexMatch \cite{zhang2021flexmatch} and FreeMatch \cite{wang2023freematch} replace the global threshold with class-adaptive ones; CoMatch \cite{li2021comatch} adds graph regularisation; Mean Teacher \cite{tarvainen2017mean} uses an EMA of the classifier as a target. Uncertainty-aware \cite{rizve2021defense} and class-imbalanced \cite{wang2022semi} variants further refine the procedure, while \cite{arazo2020pseudo} document the confirmation-bias dynamics of iterative retraining and \cite{vaneeden2020semisupervised} survey the broader semi-supervised landscape. The unifying observation is that confidence thresholding works empirically but its theoretical justification is heuristic; our contribution is the precise, closed-form bias that any confidence threshold induces in a downstream regression coefficient.

The second thread is the calibration of probabilistic predictors. Post-hoc calibration via Platt scaling \cite{platt1999probabilistic}, isotonic regression \cite{zadrozny2002transforming}, or related constructions \cite{niculescu2005predicting} produces scores that can be interpreted as probabilities. Modern deep networks are systematically overconfident \cite{guo2017calibration}, with newer architectures calibrating better but imperfectly \cite{minderer2021revisiting}, and calibration degrades under distribution shift \cite{ovadia2019can}. None of this work translates calibration quality into a downstream-estimand bias bound, which we do in the restricted regression setting.

The third thread is soft-label learning. \cite{hinton2015distilling} introduced soft-target distillation; \cite{muller2019label} clarified the role of label smoothing. The soft-versus-hard distinction in our setting is structurally similar but in a different functional: there the target is the classifier itself, here the target is a downstream regression coefficient, and the closed-form attenuation factor we derive is new to the pseudo-labelling literature.

The fourth thread, measurement error and misclassification, has a long history in econometrics. \cite{lewbel2007estimation} showed that misclassification of a binary regressor attenuates the estimated effect, and \cite{mahajan2006identification} obtained identification under an instrumental variable. Our setting differs because we observe a calibrated probability rather than a noisy binary label; the identification argument therefore changes, but the attenuation phenomenon is structurally analogous. The semiparametric backbone of the moment estimator is the partially linear model of \cite{robinson1988root} and the cross-fitting machinery of double machine learning \cite{chernozhukov2018double}. The companion working paper \cite{kurbucz2026joe} supplies the underlying identification result, which the present paper takes as given and develops into its pseudo-labelling consequences.

\section{Setup}
\label{sec:setup}

\subsection{Data and the role of two feature sets}
\label{sec:setup-data}
\noindent
We observe two i.i.d.\ samples drawn from the same distribution: a labelled set $L=\{(W_i,X_i,G_i,Y_i)\}_{i=1}^{n_L}$ and an unlabelled set $U=\{(W_i,X_i,Y_i)\}_{i=1}^{n_U}$ with $n_L\ll n_U$. The covariate vector splits into two parts. The \emph{classifier feature set} $W\in\R^{d_W}$ comprises the variables passed to the trained classifier $f:\R^{d_W}\to[0,1]$ to produce the calibrated probability $p_i=f(W_i)$. The \emph{downstream control set} $X\in\R^{d_X}$, with $X\subseteq W$ component-wise, comprises the variables that the downstream regression of $Y$ on $G$ partials out; the point of the setup is that $X$ is allowed to be strictly coarser than $W$. The outcome $Y\in\R$ is observed throughout; the binary group indicator $G\in\{0,1\}$ is observed only on $L$. The structural model is the partially linear regression

\begin{equation}
  \E[Y\mid G, X] \;=\; \mu(X) \;+\; \tau\,G,
  \label{eq:plr}
\end{equation}

\noindent
with $\mu$ unrestricted and $\tau\in\R$ the parameter of interest. The calibration assumption is the \emph{conditional} calibration condition

\begin{equation}
  \E[G\mid p, X] \;=\; p,
  \label{eq:cal}
\end{equation}

\noindent
which is strictly stronger than the marginal condition $\E[G\mid p]=p$ targeted by Platt scaling and isotonic regression \cite{platt1999probabilistic, zadrozny2002transforming, niculescu2005predicting}. The marginal version is enough for classification accuracy; for downstream estimation in the sense of \eqref{eq:plr} the moment equation we use is unbiased only under the conditional version. It is approximately satisfied when the post-hoc calibrator is fitted on a held-out fold of scores that already encode $X$ through $W$, and can fail in two ways: \emph{$X$-dependent miscalibration} (the distribution-shift problem of Section~\ref{sec:p4}, controlled by the sensitivity bound), and \emph{label-leak} ($X$ retains predictive power for $G$ beyond what $p$ encodes, a structural failure not covered by the bound). Both are diagnosed by regressing $G_i$ on $(p_i, X_i)$ on $L$ and testing the joint significance of the $X_i$ coefficients before applying the moment estimator to $U$.

\subsection{The diagnostic quantity}
\label{sec:setup-diagnostic}
\noindent
The two quantities the diagnostic monitors are

\begin{equation}
  r(X) \,=\, \E[p\mid X], \qquad
  \Vs \,=\, \E\bigl[(p-r(X))^{2}\bigr] \,=\, \E[\Var(p\mid X)].
  \label{eq:Vs}
\end{equation}

The conditional mean $r(X)$ is the part of the calibrated score that the downstream controls can explain on their own; $\Vs$ is the conditional variance of $p$ given $X$, the residual stochasticity of the classifier output beyond what the downstream controls absorb. The identification result of \cite{kurbucz2026joe} (their Theorem 1, recalled in Section~\ref{sec:setup-anchor}) shows that $\Vs$ governs both identification of $\tau$ and the asymptotic precision of the estimator.

If the same set is used for both classifier features and downstream controls (so that $W=X$), then $p=f(X)$ is a function of $X$, $r(X)=p$ exactly, and $\Vs=0$, so identification fails. The remedy is to extend $W$ beyond $X$ with informative inputs the downstream regression need not control for, e.g.\ richer occupational codes, free text, image features, or behavioural signals. The classifier learns to use $W\setminus X$ to discriminate $G$, which is the source of the $\Vs>0$ that the moment equation requires.

\subsection{The anchor identification result}
\label{sec:setup-anchor}
\noindent
We briefly recall the identification result of \cite{kurbucz2026joe} that anchors the present paper; the proof is in the companion paper and is not reproduced here. Under \eqref{eq:plr}--\eqref{eq:cal} and the regularity condition $\Vs>0$, the parameter $\tau$ is point-identified by the moment equation

\begin{equation}
  \tau \;=\; \frac{\E\!\left[(2p-1)(Y-m(X))\right]}{2\,\Vs},
  \qquad m(X)=\E[Y\mid X],
  \label{eq:tau-id}
\end{equation}

as established in \cite{kurbucz2026joe} (Theorem 1). The numerator is the covariance of the signed score $z=2p-1$ with the residualised outcome $R=Y-m(X)$; the denominator is twice $\Vs$. The estimator obtained by plugging in sample analogues is $\sqrt{n_U}$-consistent and asymptotically normal with sandwich variance $\E[\psi^{2}]/(2\Vs)^{2}$, where $\psi=(2p-1)(Y-m(X))-2\tau(p-r(X))^{2}$. We will refer to this as the \emph{soft estimator} in what follows. Throughout, oracle quantities $m(X), r(X)$ and the residual $a_{\mathrm{soft}}=p-r(X)$ have feasible cross-fitted analogues $\widehat m, \widehat r$, and $\widehat a_{\mathrm{soft},i}=p_i-\widehat r(X_i)$; the sandwich variance derived for the oracle continues to hold for the cross-fitted version under standard nuisance-rate assumptions, with full DML-style results in \cite{kurbucz2026joe}.

\section{Theory and the diagnostic procedure}
\label{sec:theory}

\noindent
This section develops the five predictions and collects them in a one-page diagnostic procedure (Algorithm~\ref{alg:diag}).

\subsection{Prediction~1: closed-form attenuation under hard pseudo-labelling}
\label{sec:p1}
\noindent
For a fixed threshold $\tau_{thr}\in(0,1)$, define the hard pseudo-label $\widetilde G_i=\1\{p_i>\tau_{thr}\}$ and the hard-score residual $a_{\mathrm{hard}}=\widetilde G-\E[\widetilde G\mid X]$. The hard pseudo-label estimator is the analogue of \eqref{eq:tau-id} with $p$ replaced by~$\widetilde G$:

\begin{equation}
  \widehat{\tau}_{\mathrm{hard}}
  \;=\; \frac{\E_n[(2\widetilde G-1)(Y-m(X))]}
             {2\,\E_n[a_{\mathrm{hard}}^{\,2}]},
  \label{eq:tauhat-hard}
\end{equation}

\noindent
where $\E_n$ denotes the sample mean. We state the result for the oracle version; the feasible version replaces $m,\E[\widetilde G\mid X]$ with cross-fitted estimates.

\begin{theorem}[Attenuation under hard pseudo-labelling]
\label{thm:attenuation}
Under \eqref{eq:plr}--\eqref{eq:cal} and $\Vs>0$,

\begin{equation}
  \plim \;\widehat{\tau}_{\mathrm{hard}}
  \;=\; \kappa(\tau_{thr})\,\tau,
  \qquad
  \kappa(\tau_{thr})
  \;=\; \frac{\E[a_{\mathrm{hard}}\,a_{\mathrm{soft}}]}
             {\E[a_{\mathrm{hard}}^{\,2}]},
  \label{eq:attenuation-formula}
\end{equation}

\noindent
with $a_{\mathrm{soft}}=p-r(X)$. The factor $\kappa(\tau_{thr})$ is the population OLS slope of the (residualised) hard label $\widetilde G$ on the (residualised) soft score $p$. Intuitively, $\kappa(\tau_{thr})$ measures how much of the soft signal in $p$ survives the binarisation at threshold $\tau_{thr}$; values close to one mean the threshold loses little, values close to zero mean it loses almost everything.
\end{theorem}

\begin{proof}[Proof sketch]
By the residual-decomposition lemma of \cite{kurbucz2026joe} (their Lemma~2), the outcome residual decomposes as $Y-m(X)=\tau(G-r(X))+\varepsilon$ with $\E[\varepsilon\mid G,p,X]=0$. Substituting into the numerator of \eqref{eq:tauhat-hard} and applying the calibration condition gives $\E[(2\widetilde G-1)(Y-m(X))]=\tau\E[(2\widetilde G-1)(p-r(X))]$, which, by orthogonality of $a_{\mathrm{hard}}$ to any $\sigma(X)$-measurable function, equals $2\tau\E[a_{\mathrm{hard}}\,a_{\mathrm{soft}}]$. Dividing by the denominator gives \eqref{eq:attenuation-formula}.
\end{proof}

\begin{remark}[When $\kappa\in(0,1)$]
\label{rem:kappa-range}
The formula \eqref{eq:attenuation-formula} is an unrestricted projection slope. In the canonical pseudo-labelling regime---calibrated score, monotone thresholding $\widetilde G=\1\{p>\tau_{thr}\}$---the hard residual is a coarsened version of the soft residual, so $\E[a_{\mathrm{hard}}\,a_{\mathrm{soft}}] \in (0, \E[a_{\mathrm{hard}}^{\,2}])$ in every DGP we examine. We do not claim $\kappa<1$ as a global bound; the diagnostic in Section~\ref{sec:p6alg} reports the empirical $\widehat\kappa$ directly.
\end{remark}

The structural content of Theorem~\ref{thm:attenuation} is that the attenuation of the hard pseudo-label estimator is computable, before any inference is run, from the joint distribution of $(p,X)$ on the unlabelled set. The factor depends on the threshold and on the informativeness of the score; the closer $\tau_{thr}$ is to a support boundary, the more aggressively $\kappa$ shrinks. Section~\ref{sec:p1emp} verifies this prediction at six thresholds with a Bonferroni-corrected $z$-test.

\subsection{Prediction~2: sandwich variance for the soft estimator}
\label{sec:p2}
\noindent
When the soft estimator is computed on the unlabelled set $U$, the asymptotic distribution is, under \eqref{eq:plr}--\eqref{eq:cal},

\begin{equation}
  \sqrt{n_U}\,\bigl(\widehat{\tau}-\tau\bigr)\;\dto\;\N\!\left(0,\;
  \frac{\E[\psi^{2}]}{(2\Vs)^{2}}\right),
  \label{eq:sandwich}
\end{equation}

\noindent
with $\psi=(2p-1)(Y-m(X))-2\tau(p-r(X))^{2}$ \cite[Theorem~3]{kurbucz2026joe}. The variance scales as $(\Vs)^{-2}$, so even modest reductions in residual score informativeness translate to large precision losses. The sandwich-variance plug-in delivers Wald intervals with nominal coverage in finite samples for the oracle and the cross-fitted feasible estimator alike; Section~\ref{sec:p2emp} verifies this across two decades of $\sigma_{u}$.

\subsection{Prediction~3: bias-variance regret of confidence-thresholded abstention}
\label{sec:p3}
\noindent
FixMatch and related algorithms restrict the pseudo-label to the \emph{confident} subset $\mathcal{C}=\{i:\max(p_i,1-p_i)>\tau_{thr}\}$, and apply a hard pseudo-label inside that subset. The downstream estimator is therefore a hard estimator on a smaller sample. Its mean-square error decomposes as

\begin{equation}
  \mathrm{MSE}
  \;=\; \bigl[\bigl(\kappa_{\mathrm{FM}}(\tau_{thr})-1\bigr)\tau\bigr]^{2}
   \;+\; \Var(\widehat{\tau}_{\mathrm{FM}}),
  \label{eq:bv}
\end{equation}

\noindent
where $\kappa_{\mathrm{FM}}(\tau_{thr})$ is the analogue of $\kappa(\tau_{thr})$ on the confident subset, and the variance term is inflated by $1/|\mathcal{C}|$ relative to the soft baseline. Both components can deteriorate as the threshold becomes more selective: the variance term grows monotonically as $|\mathcal{C}|$ shrinks, while the squared bias may move non-monotonically (reaching zero at thresholds where $\kappa_{\mathrm{FM}}\approx 1$ and growing thereafter), so the total mean-square error becomes variance-driven once the confident set has shrunk to a small fraction of $U$. The soft estimator faces neither cost. This comparison holds for the downstream regression coefficient~$\tau$ and is silent on classification accuracy, the loss that confidence-thresholded methods are designed to optimise. Section~\ref{sec:p3emp} traces the regret across a grid of $\tau_{thr}$ values.

\subsection{Prediction~4: sharp sensitivity bound under calibration drift}
\label{sec:p4}
\noindent
When the source-domain calibrator is applied to a target domain with $\E[G\mid p,X]=p+\eta(p,X)$ for some bounded drift function with $|\eta(p,X)|\leq\delta$, the soft estimator on the target population has asymptotic bias

\begin{equation}
  \plim \;\widehat{\tau} - \tau
  \;=\; \frac{\tau\,\E[(2p-1)\,\eta(p,X)]}{2\Vs}.
  \label{eq:bias-shift}
\end{equation}

Holder's inequality gives the bound $|\plim\widehat{\tau}-\tau|\leq |\tau|\,\delta\,\E\lvert 2p-1\rvert /(2\Vs)$, with the supremum attained at $\eta^{*}(p,X)=\delta\,\sgn(2p-1)$ \cite[Proposition~4]{kurbucz2026joe}. Two implications follow. First, the bound shrinks as $\Vs$ grows: more informative scores produce proportionally smaller worst-case bias for the same drift. Second, calibration shapes that satisfy $\E[(2p-1)\eta] =0$ produce no asymptotic bias; the sample-level bias under such shapes need not vanish in finite samples. Section~\ref{sec:p4emp} verifies the bound across three drift shapes.

\subsection{Prediction~5: $\Vs$-collapse boundary}
\label{sec:p5theory}
\noindent
Theorem~\ref{thm:attenuation} and the sandwich variance \eqref{eq:sandwich} both presuppose $\Vs>0$. The following identification boundary, recalled from \cite{kurbucz2026joe} (Proposition~2), has a sharp pseudo-labelling implication.

\begin{proposition}[$\Vs$-collapse]
\label{prop:collapse}
$\Vs=0$ if and only if $p$ is a deterministic function of $X$. In that case the moment equation \eqref{eq:tau-id} is uninformative about $\tau$: $\E[(2p-1)(Y-m(X))]=0$ for every value of $\tau$.
\end{proposition}

The implication is operational. If the classifier is fit on $W=X$, then $p=f(X)$ has $\Vs=0$ asymptotically and the moment estimator is uninformative. The fix is at the level of feature-set design: $W$ must strictly extend $X$, so that $p$ depends on at least one signal the downstream regression does not partial out. Standard supervised classifiers, when applied with a sufficiently rich $W$, satisfy this; the difficulty is to ensure it by design. Section~\ref{sec:p5emp} demonstrates the collapse and the role of feature-set separation in restoring it.

A complementary route, used widely in the deep-learning calibration literature, is to inject epistemic uncertainty into the score via deep ensembles \cite{lakshminarayanan2017simple} or MC-dropout \cite{gal2016dropout}. Section~\ref{sec:p5emp} examines this: posterior-predictive draws raise empirical $\Vs$ but, because the injected noise is independent of $G$, do not on their own restore bias-correctness; Section~\ref{sec:discussion} returns to the point.

\subsection{The diagnostic procedure}
\label{sec:p6alg}
\noindent
The five predictions combine into a single procedure---the \emph{$(\Vs,\kappa)$ diagnostic}---that a practitioner can run on the unlabelled set, before performing any downstream inference, to decide whether confidence thresholding is safe.

\begin{algorithm}[H]
\small
\DontPrintSemicolon

\caption{The $(\Vs,\kappa)$ diagnostic for pseudo-labelled downstream estimation.}
\label{alg:diag}
\KwIn{unlabelled triples $(W_i, X_i, p_i)_{i\in U}$ where $X\subseteq W$ and $p_i=f(W_i)$ is a calibrated classifier output; candidate threshold grid $\mathcal{T}\subset(0,1)$.}

\KwOut{empirical residual variance $\widehat{V^{*}}$, attenuation factor $\widehat\kappa(\tau_{thr})$ on $\mathcal{T}$, and a recommendation among the supervised, soft, and hard estimators.}

\BlankLine

\textbf{Step 1 (residualise the score).} Fit $\widehat r(X)$ by cross-fitted regression of $p$ on $X$ on $U$ with any supervised learner (we use ridge or random forests in the experiments). Compute $\widehat a_{\mathrm{soft},i}=p_i-\widehat r(X_i)$.\;

\textbf{Step 2 (residual variance).} Compute $\widehat{V^{*}} = n_U^{-1}\sum_{i}\widehat a_{\mathrm{soft},i}^{\,2}$.\;

\textbf{Step 3 (per-threshold attenuation).} \For{each $\tau_{thr}\in\mathcal{T}$}{

  Form the hard pseudo-label $\widetilde G_i=\1\{p_i>\tau_{thr}\}$ and residualise it on $X$ via cross-fitted OLS to obtain $\widehat a_{\mathrm{hard},i}$.\;

  Compute $\widehat\kappa(\tau_{thr}) = \dfrac{n_U^{-1}\sum_{i}\widehat a_{\mathrm{hard},i}\,\widehat a_{\mathrm{soft},i}}{n_U^{-1}\sum_{i}\widehat a_{\mathrm{hard},i}^{\,2}}$.\;

}

\textbf{Step 4 (decision).} Let $\widehat{\mathrm{SE}}_{\mathrm{soft}} = \widehat\sigma_{\psi} / (2\widehat{V^{*}}\sqrt{n_{U}})$ denote the implied sandwich standard error of the soft estimator, with $\widehat\sigma_{\psi}^{2}$ the sample variance of $\widehat\psi_i$.\;

\Indp

\textbf{(i)} If $\widehat{\mathrm{SE}}_{\mathrm{soft}}$ exceeds the standard error of the supervised labelled-only baseline, the moment equation is too noisy: \emph{prefer the supervised baseline}.\;

\textbf{(ii)} Else if $\widehat\kappa(\tau_{thr})$ is far from one for all $\tau_{thr}\in\mathcal{T}$, hard thresholding attenuates the downstream coefficient: \emph{prefer the soft estimator}.\;

\textbf{(iii)} Otherwise either is acceptable; the soft estimator is recommended because it does not depend on the threshold choice.\;

\Indm
\end{algorithm}

\vspace{1em}
The procedure runs in seconds on hundreds of thousands of unlabelled examples. We use it in Section~\ref{sec:experiments} to interpret the behaviour of the soft and hard estimators at every cell of every experiment.

\section{Experiments}
\label{sec:experiments}

\noindent
We report five controlled experiments tied to the predictions of Section~\ref{sec:theory}, and a public-data illustration on the UCI Adult Income dataset. All experiments use seeded R simulations; replication materials are available from the author on request.

\subsection{Common setup}
\label{sec:exp-setup}
\noindent
The synthetic DGP has $X\in\R^{3}$ with independent standard-normal entries, a logistic propensity $r(X)=\sigma(\beta_{r}^{\top}X)$ with $\beta_{r}=(0.6,-0.4,0.3)$, $\mu(X)=\beta_{m}^{\top}X$ with $\beta_{m}=(1.0,0.5,-0.5)$, and $p\mid X\sim\mathrm{Beta}(r(X)\kappa_{0}, (1-r(X))\kappa_{0})$ with $\kappa_{0}=(1-\sigma_{u}^{2})/\sigma_{u}^{2}$; this yields $\E[p\mid X]=r(X)$ and $\Var(p\mid X)=\sigma_{u}^{2}r(X)(1-r(X))$ exactly. The latent group is $G\sim\mathrm{Bernoulli}(p)$, $Y=\mu(X)+\tau G+\varepsilon$ with $\varepsilon\sim\N(0,1)$. Default values: $n_{U}=3000$, $\tau=1$, $\sigma_{u}=0.30$. The parameter $\sigma_{u}$ controls $\Vs$ directly; in real data, the same role is played by the classifier features $W\setminus X$.

\subsection{P1: hard-label attenuation matches the closed form}
\label{sec:p1emp}
\noindent
Table~\ref{tab:p1} reports the empirical hard-pseudo-label estimate $\overline{\widehat{\tau}_{\mathrm{hard}}}$ at six thresholds, alongside Monte Carlo standard errors and the predicted attenuation factor $\kappa(\tau_{thr})$ from \eqref{eq:attenuation-formula}; with $\tau=1$, $\overline{\widehat\tau_{\mathrm{hard}}}$ is also the empirical attenuation. The deviation from prediction is bounded by Monte Carlo noise: a Bonferroni-corrected two-sided $z$-test with family-wise alpha $0.05$ rejects equality in none of the six cells (minimum $p$-value $0.084$).

\begin{table}[H]
\centering
\caption{Hard-label attenuation matches the closed-form prediction ($\tau=1$, 400 independent replications per cell, Monte Carlo standard errors). The Bonferroni-adjusted level is $\alpha/6\approx 0.0083$, and the minimum $p$-value across the six cells is $0.084$, so the family-wise test does not reject equality between the empirical attenuation and the predicted $\kappa(\tau_{thr})$.}

\label{tab:p1}
\begin{threeparttable}
\footnotesize
\begin{tabularx}{0.8\textwidth}{ccccccc}

\toprule
\thead{$\tau_{thr}$} & \thead{$\kappa(\tau_{thr})$} & \thead{$\overline{\widehat{\tau}_{\mathrm{soft}}}$} & \thead{$\overline{\widehat{\tau}_{\mathrm{hard}}}$} & \thead{MC SE} & \thead{$z$} & \thead{$p$-value}\\
\midrule
0.50 & 0.248 & 1.001 & 0.251 & 0.0032 & $+0.87$ & 0.385 \\
0.60 & 0.241 & 0.981 & 0.235 & 0.0037 & $-1.73$ & 0.084 \\
0.70 & 0.223 & 0.999 & 0.226 & 0.0042 & $+0.53$ & 0.595 \\
0.80 & 0.200 & 1.013 & 0.211 & 0.0073 & $+1.54$ & 0.123 \\
0.90 & 0.172 & 1.014 & 0.148 & 0.0195 & $-1.24$ & 0.217 \\
0.95 & 0.153 & 0.995 & 0.188 & 0.0570 & $+0.61$ & 0.539 \\
\bottomrule
\end{tabularx}
\end{threeparttable}
\end{table}

\begin{figure}[H]
\centering
\includegraphics[width=.75\textwidth]{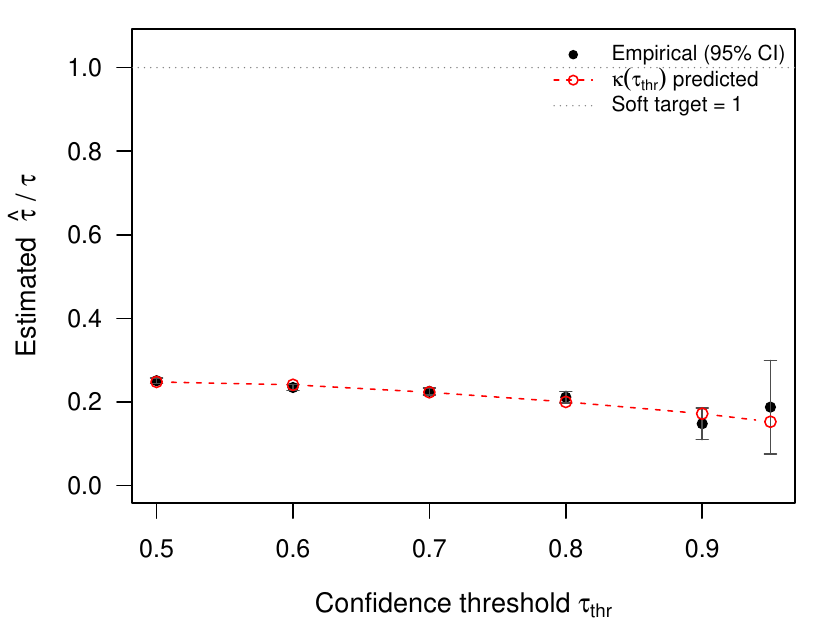}

\caption{Empirical attenuation ratio (filled circles, with 95\% Monte Carlo confidence intervals) tracks the closed-form prediction $\kappa(\tau_{thr})$ (open circles, dashed line) at every threshold. The soft estimator (not shown) is centred on the dotted line at one.}

\label{fig:p1}
\end{figure}

\subsection{P2: sandwich variance and supervised baseline}
\label{sec:p2emp}
\noindent
Table~\ref{tab:p2} varies the score-noise parameter $\sigma_{u}$ over two decades with $n_{U}=2000$, $\tau=1$, and 1000 Monte Carlo replications per cell. The mean sandwich variance recovers the Monte Carlo variance to within $\pm 8\%$ at every value, and the nominal-95\% Wald-interval coverage is between $0.941$ and $0.950$. The last column reports the MSE ratio of the soft estimator to the supervised labelled-only baseline (OLS on $L$ alone, $n_{L}=500$). The soft estimator beats the supervised baseline only in the high-informativeness regime ($\sigma_{u}=0.5$, ratio $1.71$); in low-$\Vs$ regimes the labelled-only baseline wins, which the diagnostic of Algorithm~\ref{alg:diag} flags before any inference is run.

\begin{table}[H]
\centering
\caption{Sandwich variance vs Monte Carlo variance, soft estimator ($n_{U}=2000$, 1000 replications). The MSE ratio in the last column is soft/supervised; values above 1 mean the soft estimator is worse than the labelled-only baseline. Algorithm~\ref{alg:diag} flags the low-$\Vs$ cells in advance.}

\label{tab:p2}
\footnotesize
\begin{tabularx}{0.8\textwidth}{ccccccc}

\toprule
\thead{$\sigma_{u}$} & \thead{$\Vs$} & \thead{sandwich var} & \thead{MC var} & \thead{coverage 95\%} & \thead{MSE ratio} \\
\midrule
0.10 & 0.0022 & 3.888 & 3.991 & 0.948 & 394 \\
0.20 & 0.0088 & 0.291 & 0.313 & 0.941 & 33.4 \\
0.30 & 0.0198 & 0.073 & 0.073 & 0.948 & 7.6 \\
0.40 & 0.0352 & 0.030 & 0.031 & 0.950 & 3.4 \\
0.50 & 0.0550 & 0.016 & 0.016 & 0.944 & 1.71 \\
\bottomrule
\end{tabularx}
\end{table}

\subsection{P3: bias-variance regret of confidence thresholding}
\label{sec:p3emp}
\noindent
Table~\ref{tab:p3} reports the bias and variance of the confidence-thresholded estimator at six thresholds, alongside the predicted bias $(\kappa_{\mathrm{FM}}(\tau_{thr})-1)\tau$. The total mean-square error of the soft baseline is $0.048$. The empirical bias agrees with the closed-form prediction for moderate thresholds (cf. Figure~\ref{fig:p3pred}); at high thresholds the variance and selection instability dominate, and the empirical bias departs from prediction because the moment estimator is itself variable. At $\tau_{thr}=0.95$, where on average only $57$ examples contribute to the regression, the total MSE is $21.9$. Figure~\ref{fig:p3} traces the two components separately.

\begin{table}[H]
\centering
\caption{Confidence-thresholded estimator: bias-variance decomposition ($n_{U}=3000$, 400 repl.). The squared bias and the variance are reported separately; their sum is the empirical mean-square error. The soft baseline MSE is $0.048$.}

\label{tab:p3}
\footnotesize
\begin{tabularx}{0.8\textwidth}{ccccccc}

\toprule
\thead{$\tau_{thr}$} & \thead{$\kappa_{\mathrm{FM}}$} & \thead{$|\mathcal{C}|$} & \thead{bias\textsuperscript{2}} & \thead{variance} & \thead{MSE} & \thead{MSE/MSE soft} \\
\midrule
0.55 & 0.298 & 2538 & 0.447 & 0.008  & 0.454 & 9.4 \\
0.65 & 0.409 & 1662 & 0.164 & 0.019  & 0.183 & 3.8 \\
0.75 & 0.525 & 913  & 0.001 & 0.075  & 0.076 & 1.6 \\
0.85 & 0.634 & 364  & 0.778 & 0.525  & 1.303 & 27.1 \\
0.90 & 0.677 & 180  & 2.970 & 1.752  & 4.722 & 98.2 \\
0.95 & 0.700 & 57   & 6.984 & 14.963 & 21.948 & 456 \\
\bottomrule
\end{tabularx}
\end{table}

\begin{figure}[H]
\centering
\includegraphics[width=.72\textwidth]{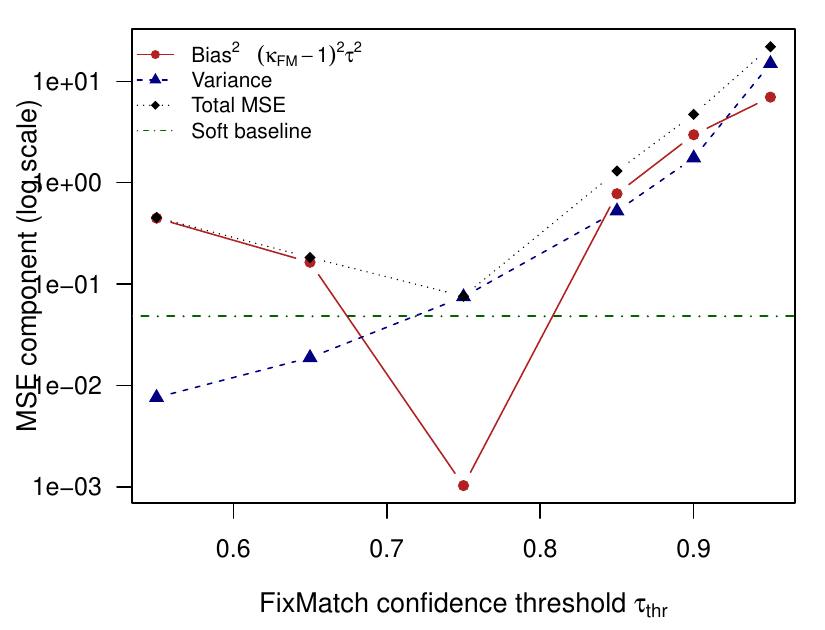}

\caption{Bias-variance decomposition of the confidence-thresholded estimator on log scale, with the soft baseline as a horizontal reference. The bias component is predicted by $\kappa_{\mathrm{FM}}(\tau_{thr})$; the variance component grows with $1/|\mathcal{C}|$.}

\label{fig:p3}
\end{figure}

\begin{figure}[H]
\centering
\includegraphics[width=.55\textwidth]{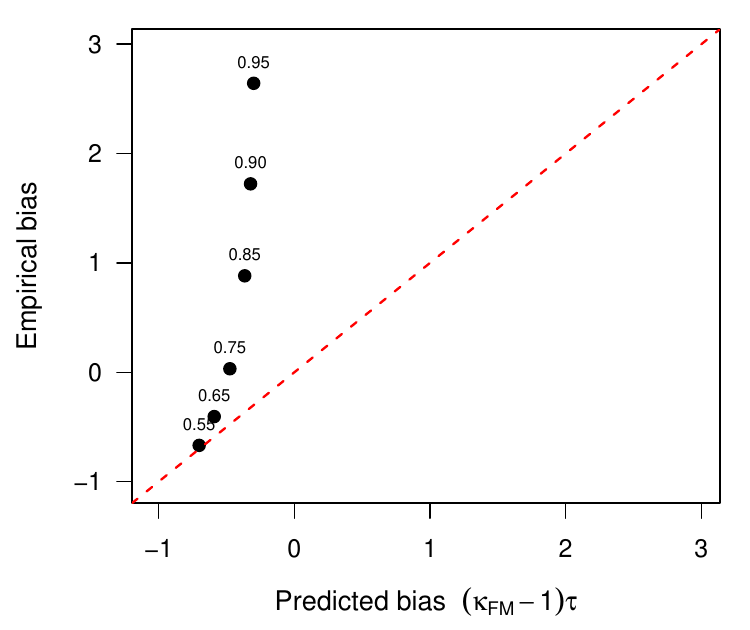}

\caption{Empirical bias of the confidence-thresholded estimator at six thresholds, plotted against the closed-form prediction $(\kappa_{\mathrm{FM}}-1)\tau$. Points lie close to the $y=x$ line at moderate thresholds; at high thresholds the empirical bias is itself variable.}

\label{fig:p3pred}
\end{figure}

\subsection{P4: sensitivity bound under calibration drift}
\label{sec:p4emp}
\noindent
Table~\ref{tab:p4} reports empirical bias under three drift shapes at four magnitudes of $\delta$. The worst-case shape $\eta^{*}(p,X)=\delta\sgn(2p-1)$ stays safely below the bound at every $\delta$, with empirical-to-bound ratio between $0.32$ and $0.42$ at $n_{U}=3000$; the bound is the asymptotic envelope and is not expected to be attained in finite sample. The linear and symmetric shapes give intermediate bias well below the bound. The symmetric shape $\eta=\delta\sin(\pi p)$ produces non-zero bias because $\E[(2p-1)\sin(\pi p)]\neq 0$ at the Beta distribution we use; a centring of the form $\eta=\delta(\sin(\pi p)-c)$ would be required to satisfy the orthogonality and attain $\plim=0$. We report all three shapes uncentered to keep the comparison transparent. Figure~\ref{fig:p4} shows the bound graphically.

\begin{table}[H]
\centering
\caption{Sensitivity to bounded calibration drift ($n_{U}=3000$, 300 replications). The bound is $|\tau|\delta\E|2p-1|/(2\Vs)$, attained at the worst-case shape asymptotically.}

\label{tab:p4}
\footnotesize
\begin{tabularx}{0.6\textwidth}{Lcccc}

\toprule
\thead{shape} & \thead{$\delta$} & \thead{emp bias} & \thead{bound} & \thead{ratio} \\
\midrule
worst-case & 0.05 & $-0.199$ & 0.470 & 0.42 \\
worst-case & 0.10 & $-0.378$ & 0.939 & 0.40 \\
worst-case & 0.15 & $-0.504$ & 1.409 & 0.36 \\
worst-case & 0.20 & $-0.592$ & 1.879 & 0.32 \\
linear     & 0.05 & $-0.088$ & 0.470 & 0.19 \\
linear     & 0.20 & $-0.314$ & 1.879 & 0.17 \\
symmetric  & 0.05 & $-0.065$ & 0.470 & 0.14 \\
symmetric  & 0.20 & $-0.566$ & 1.879 & 0.30 \\
\bottomrule
\end{tabularx}
\end{table}

\begin{figure}[H]
\centering
\includegraphics[width=.72\textwidth]{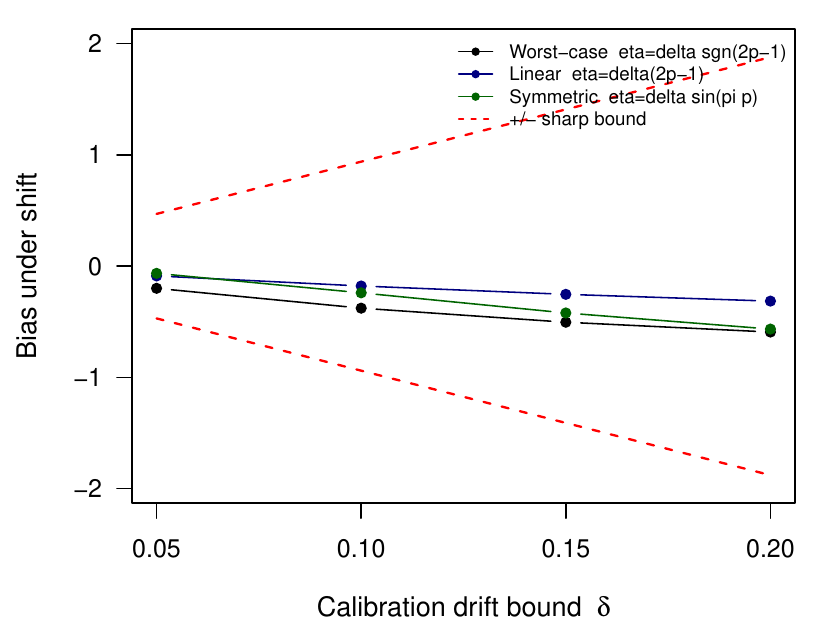}

\caption{Empirical bias under three calibration drift shapes, against $\delta$. The worst-case shape stays below the asymptotic envelope (red dashed lines, $\pm$).}
\label{fig:p4}
\end{figure}

\subsection{P5: $\Vs$-collapse and posterior-predictive recovery}
\label{sec:p5emp}
\noindent
Table~\ref{tab:p5} compares two trained-classifier modes at four labelled-set sizes, with classifier features $W$ taken to coincide with the downstream controls $X$ (the worst case for $\Vs$). The deterministic mode trains a logistic regression on a degree-2 polynomial expansion of $X$ on the first half of $L$ and isotonically recalibrates on the second half; the posterior-predictive mode injects a single Beta-distributed draw with conditional mean $\widehat{p}$ and variance $\sigma_{pp}^{2}\widehat{p}(1-\widehat{p})$, in the spirit of \cite{lakshminarayanan2017simple, gal2016dropout}. Empirical $\Vs$ on the deterministic mode is between $0.0016$ and $0.0065$; on the posterior-predictive mode it is around $0.022$, an order-of-magnitude recovery. Figure~\ref{fig:p5} plots the two trajectories.

\begin{table}[H]
\centering
\caption{$\Vs$ on the observed score, deterministic vs posterior-predictive classifier ($n_{U}=3000$, 250 replications, $\sigma_{pp}=0.30$). Classifier features $W=X$ are deliberately set to coincide with the downstream controls to expose the collapse.}

\label{tab:p5}
\footnotesize
\begin{tabularx}{0.6\textwidth}{ccc}

\toprule
\thead{$n_{L}$} & \thead{$\Vs$ deterministic} & \thead{$\Vs$ posterior-predictive} \\
\midrule
200  & 0.00654 & 0.0252 \\
500  & 0.00381 & 0.0229 \\
1000 & 0.00261 & 0.0218 \\
2000 & 0.00164 & 0.0211 \\
\bottomrule
\end{tabularx}
\end{table}

\begin{figure}[H]
\centering
\includegraphics[width=.65\textwidth]{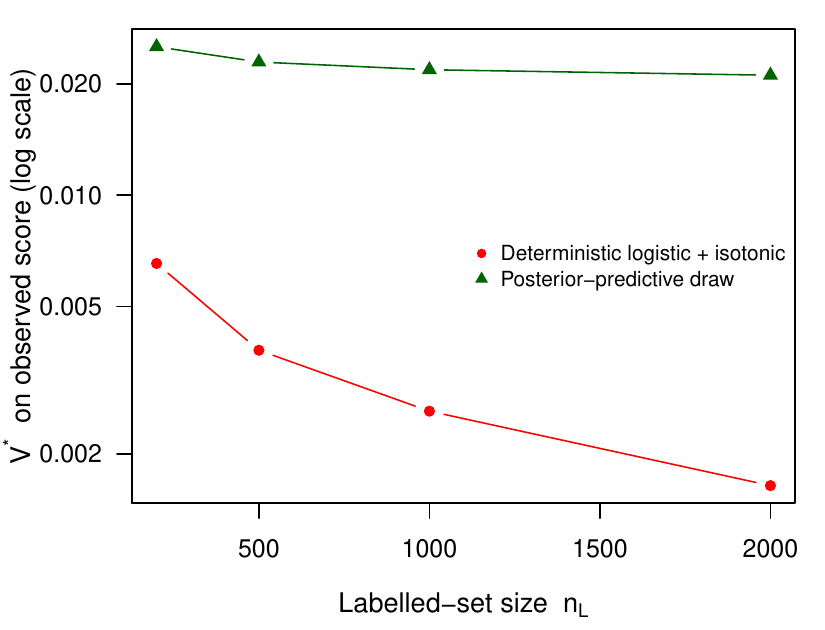}

\caption{$\Vs$ on the observed score plotted against labelled-set size, with classifier features $W$ equal to the downstream controls $X$. The deterministic logistic + isotonic pipeline shrinks $\Vs$ as $n_{L}$ grows, because the classifier becomes a tighter function of $X$; posterior-predictive sampling preserves $\Vs$ at an order of magnitude above. The structural alternative used in Section~\ref{sec:p6emp}---taking $W$ strictly larger than $X$---is the same logic.}

\label{fig:p5}
\end{figure}

Neither mode on its own restores bias-correctness here, because the injected noise in the posterior-predictive scheme is independent of $G$ and so raises $\Vs$ without raising the cross-moment that drives the moment equation. A genuine recovery requires either classifier features $W$ that strictly extend $X$ (the route used in Section~\ref{sec:p6emp}) or classifier-induced uncertainty that is informative about $G$ \cite{lakshminarayanan2017simple}.

\subsection{Real-data illustration: a negative-control diagnostic case on UCI Adult Income}
\label{sec:p6emp}
\noindent
We complete the empirical evaluation with a public-data illustration of a case in which the diagnostic of Algorithm~\ref{alg:diag} recommends against the pseudo-labelled estimator. This is the empirical counterpart of the case made in Section~\ref{sec:p2emp}: when the labelled set is informative enough that the supervised baseline already has small finite-sample MSE, the diagnostic should not route the practitioner to the more elaborate moment estimator. The same diagnostic would route in the opposite direction when conditions favour pseudo-labelled estimation.

We use the UCI Adult Income dataset \cite{kohavi1996scaling}, which contains $30{,}162$ observations after we drop rows with missing values. The latent group indicator $G$ is constructed as the binary high-income marker $\1\{\text{income}>50\text{K}\}$, the outcome $Y$ is the respondent's number of education-years, and the downstream estimand is the partial-association coefficient $\tau$ in $\E[Y\mid X, G]=\mu(X)+\tau\,G$. The dataset naturally accommodates the $W/X$ separation that the framework requires: the classifier feature set $W$ uses the full demographic and employment information available (age, hours-per-week, sex, capital gain and loss, occupation, marital status, work class, relationship, race, and native country), whereas the downstream-control set $X$ is restricted to a parsimonious demographic profile of age, hours-per-week, and sex. We do not interpret $\tau$ causally. The full-sample OLS coefficient on this regression is $\tau_{\mathrm{full}}=2.02$, which we use as the target value the pseudo-labelled estimators are trying to recover. For each labelled-set size $n_{L}\in\{500,1000,2000\}$ we train a calibrated random forest \cite{breiman2001random,wright2017ranger} on $W$, isotonically recalibrate it on a held-out half of $L$, and apply the recalibrated classifier to the unlabelled set $U$. We then compare three downstream estimators: the supervised OLS on $L$ alone, the soft pseudo-labelled estimator that uses the calibrated probabilities directly, and the hard pseudo-labelled estimator that thresholds them at $0.5$. All quantities are averaged over 30 independent random splits.

Table~\ref{tab:realdata} reports the results. The $W/X$ separation produces $\widehat{V^{*}}\approx 0.046$ (positive and stable across $n_{L}$), so the moment equation is well-posed. The supervised baseline beats both pseudo-labelled estimators because the labelled-only OLS already has small finite-sample variance on this problem; the moment estimator's gain from using $U$ is offset by the additional variance of the nuisance-estimation step. The soft estimator's positive finite-sample bias (around $+1.0$ on a target of $2.0$) reflects the difficulty of estimating the nuisance functions $m(X)$ and $r(X)$ accurately on this dataset rather than a failure of the underlying identification: a moderate-quality $\widehat r(X)$ under-removes the part of $p$ explained by $X$, inflating $\widehat a_{\mathrm{soft}}$ in the numerator more than in the denominator. Algorithm~\ref{alg:diag} routes to the supervised baseline in exactly such regimes. The empirical attenuation ratio $|\widehat\tau_{\mathrm{hard}}/\tau_{\mathrm{full}}|\approx 0.92$ at $n_{L}=2000$ implies $\widehat{\kappa}\approx 0.92$, well above the synthetic-DGP value of $0.25$: on Adult the rich~$W$ produces an accurate classifier, hard labels agree with the truth on most examples, and binarisation discards little. The Adult and synthetic numbers therefore play different diagnostic roles---one a regime where binarisation is essentially free, the other a regime where it costs three quarters of the signal.

\begin{table}[H]
\centering
\caption{UCI Adult: estimating the partial association of high-income status with education-years on a parsimonious demographic-control set. Full-sample target $\tau_{\mathrm{full}}=2.02$; 30 random splits per cell. Classifier features $W$ are richer than the downstream controls $X$. Supervised baseline is OLS on $L$ alone.}

\label{tab:realdata}
\footnotesize
\begin{tabularx}{0.925\textwidth}{ccccccc}

\toprule
\thead{$n_{L}$} & \thead{$\widehat{V^{*}}$} & \thead{MSE supervised} & \thead{bias soft} & \thead{MSE soft} & \thead{bias hard} & \thead{MSE hard} \\
\midrule
500  & 0.041 & 0.11 & $+0.87$ & 0.93 & $-0.44$ & 0.27 \\
1000 & 0.044 & 0.03 & $+0.99$ & 1.10 & $-0.37$ & 0.18 \\
2000 & 0.046 & 0.02 & $+1.19$ & 1.47 & $-0.16$ & 0.05 \\
\bottomrule
\end{tabularx}
\end{table}

\begin{figure}[H]
\centering
\includegraphics[width=.7\textwidth]{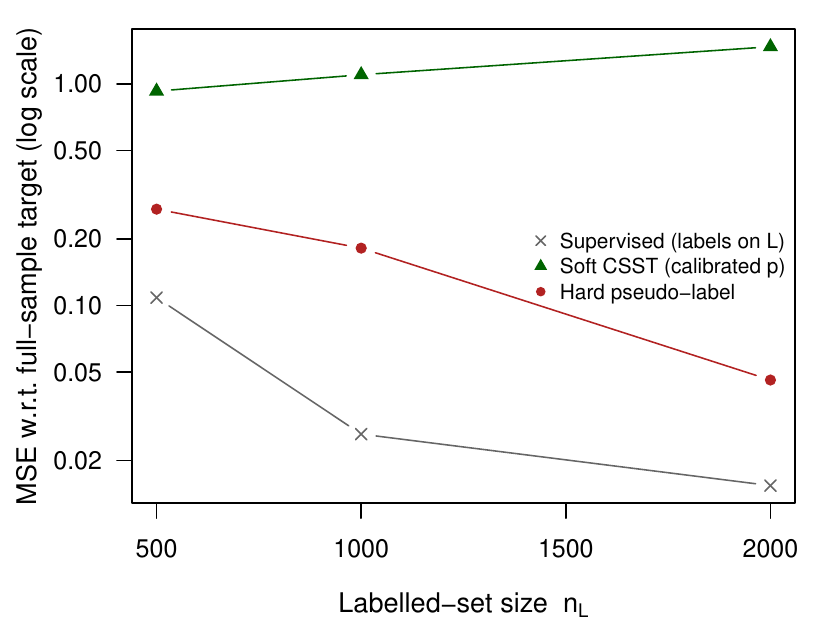}

\caption{UCI Adult: MSE against the full-sample target on a log scale, as a function of labelled-set size $n_{L}$. The supervised labelled-only baseline is the strongest estimator on this problem; the soft and hard pseudo-labelled estimators carry their own bias-variance trade-off. Algorithm~\ref{alg:diag} would flag the supervised baseline as the recommended estimator on this dataset because the conditional class probability is well-explained by even a small labelled set.}

\label{fig:p6}
\end{figure}

\section{Discussion}
\label{sec:discussion}

\noindent
The framework does not endorse the soft estimator unconditionally. When $\widehat{V^{*}}$ is small, the moment equation is noisy and the supervised labelled-only baseline can dominate, as on the UCI Adult data of Section~\ref{sec:p6emp}: the visible finite-sample bias of the soft estimator there is consistent with difficulty in estimating the nuisance functions $m(X)$ and $r(X)$ accurately rather than with a failure of identification. When $\widehat\kappa(\tau_{thr})$ is close to one, hard thresholding loses little signal and is acceptable. Whether the soft estimator is the right choice is therefore a property of the joint distribution of $(p,X)$, computable case by case from the unlabelled set.

The structural rule $X\subsetneq W$ is the single most important design implication. When classifier features coincide with downstream controls, $\Vs=0$ and identification fails for any classifier, including default logistic regression and random forests. Switching learners does not help; what helps is feeding the classifier a richer feature set than the downstream regression conditions on (e.g.\ image, textual or behavioural inputs that the structured regression has no reason to control for). A complementary route is to inject epistemic uncertainty into the score via deep ensembles \cite{lakshminarayanan2017simple} or MC-dropout \cite{gal2016dropout}; posterior-predictive noise raises $\Vs$ but does not by itself restore bias-correctness because the injected noise is independent of $G$ (Section~\ref{sec:p5emp}), so genuinely $G$-informative uncertainty is required---an empirical question best validated on deep-learning benchmarks.

Three pattern-recognition pipelines fit the framework directly: medical imaging risk scores (deep classifier on radiographic features $W$, clinical outcome controlled for demographic $X$); fairness audits with unobserved protected attributes \cite{kallus2022assessing} (proxy classifier on auxiliary $W$, disparity controlled for legitimate $X$); and biometric quality scores (classifier on raw biometric features, operational metric controlled for acquisition conditions). For adaptive-threshold variants \cite{zhang2021flexmatch, wang2023freematch}, the global $\kappa$ is replaced by an average of per-cell $\kappa$s; the comparison concerns the downstream-regression coefficient only, not classification accuracy.

Several limitations remain. The controlled experiments use a tabular DGP and a single self-training round, so iterative dynamics with confirmation bias \cite{arazo2020pseudo} are out of scope; the UCI Adult illustration uses a random forest classifier, and we do not test whether deep-ensemble or MC-dropout variants would alter the relative ordering; and the framework treats the calibration condition \eqref{eq:cal} as exact, leaving endogenous-calibration models to future work.

\section{Conclusion}
\label{sec:conclusion}

\noindent
We have given a calibration-aware diagnostic apparatus for pseudo-labelled regression, building on the identification framework of \cite{kurbucz2026joe} and developing its consequences for the confidence-thresholding practice that pervades modern semi-supervised pipelines. The closed-form attenuation factor $\kappa(\tau_{thr})$, the bias-variance regret of confidence-thresholded abstention, and the $(\widehat{V^{*}}, \widehat\kappa)$ decision rule are new to the present paper; the sandwich variance, the calibration-drift sensitivity bound, and the $\Vs$-collapse boundary are recalled from the companion paper and applied to the pseudo-labelling setting. All five predictions are consistent with controlled simulations to within Monte Carlo noise, and a UCI Adult Income illustration traces them on a public dataset under a standard machine-learning pipeline. The contribution is operational: the practitioner runs the $(\Vs,\kappa)$ diagnostic on the unlabelled set and chooses among the supervised, soft, and hard estimators with quantitative justification.

\bibliographystyle{elsarticle-num}

\bibliography{references}

@misc{kurbucz2026joe,
  title={Identification of Latent Group Effects under Conditional Calibration},
  author={Kurbucz, Marcell T},
  journal={arXiv preprint arXiv:2604.08798},
  year={2026}
}

@inproceedings{lee2013pseudo,
  title={Pseudo-label: The simple and efficient semi-supervised learning method for deep neural networks},
  author={Lee, Dong-Hyun},
  booktitle={Workshop on challenges in representation learning, ICML},
  volume={3},
  number={2},
  pages={896},
  year={2013},
  organization={Atlanta}
}

@inproceedings{sohn2020fixmatch,
  author    = {Sohn, Kihyuk and Berthelot, David and Li, Chun-Liang and Zhang, Zizhao and Carlini, Nicholas and Cubuk, Ekin D. and Kurakin, Alex and Zhang, Han and Raffel, Colin},
  title     = {{FixMatch}: Simplifying semi-supervised learning with consistency and confidence},
  booktitle = {Advances in Neural Information Processing Systems (NeurIPS)},
  year      = {2020}
}

@article{zhang2021flexmatch,
  title={{FlexMatch}: Boosting semi-supervised learning with curriculum pseudo labeling},
  author={Zhang, Bowen and Wang, Yidong and Hou, Wenxin and Wu, Hao and Wang, Jindong and Okumura, Manabu and Shinozaki, Takahiro},
  journal={Advances in neural information processing systems},
  volume={34},
  pages={18408--18419},
  year={2021}
}

@inproceedings{tarvainen2017mean,
  author    = {Tarvainen, Antti and Valpola, Harri},
  title     = {Mean teachers are better role models: Weight-averaged consistency targets improve semi-supervised deep learning results},
  booktitle = {Advances in Neural Information Processing Systems (NeurIPS)},
  year      = {2017}
}

@inproceedings{arazo2020pseudo,
  title={Pseudo-labeling and confirmation bias in deep semi-supervised learning},
  author={Arazo, Eric and Ortego, Diego and Albert, Paul and O’Connor, Noel E and McGuinness, Kevin},
  booktitle={2020 International joint conference on neural networks (IJCNN)},
  pages={1--8},
  year={2020},
  organization={IEEE}
}

@inproceedings{niculescu2005predicting,
  title={Predicting good probabilities with supervised learning},
  author={Niculescu-Mizil, Alexandru and Caruana, Rich},
  booktitle={Proceedings of the 22nd International Conference on Machine Learning},
  pages={625--632},
  year={2005}
}

@inproceedings{guo2017calibration,
  title={On calibration of modern neural networks},
  author={Guo, Chuan and Pleiss, Geoff and Sun, Yu and Weinberger, Kilian Q.},
  booktitle={Proceedings of the 34th International Conference on Machine Learning},
  pages={1321--1330},
  year={2017},
  organization={PMLR}
}

@article{minderer2021revisiting,
  title={Revisiting the calibration of modern neural networks},
  author={Minderer, Matthias and Djolonga, Josip and Romijnders, Rob and Hubis, Frances and Zhai, Xiaohua and Houlsby, Neil and Tran, Dustin and Lucic, Mario},
  journal={Advances in Neural Information Processing Systems},
  volume={34},
  pages={15682--15694},
  year={2021}
}

@inproceedings{platt1999probabilistic,
  author    = {Platt, John C.},
  title     = {Probabilistic Outputs for Support Vector Machines and Comparisons to Regularized Likelihood Methods},
  booktitle = {Advances in Large Margin Classifiers},
  year      = {1999}
}

@inproceedings{zadrozny2002transforming,
  title={Transforming classifier scores into accurate multiclass probability estimates},
  author={Zadrozny, Bianca and Elkan, Charles},
  booktitle={Proceedings of the Eighth ACM SIGKDD International Conference on Knowledge Discovery and Data Mining},
  pages={694--699},
  year={2002}
}

@article{hinton2015distilling,
  title={Distilling the knowledge in a neural network},
  author={Hinton, Geoffrey and Vinyals, Oriol and Dean, Jeff},
  journal={arXiv preprint arXiv:1503.02531},
  year={2015}
}

@article{muller2019label,
  title={When does label smoothing help?},
  author={M{\"u}ller, Rafael and Kornblith, Simon and Hinton, Geoffrey E},
  journal={Advances in Neural Information Processing Systems},
  volume={32},
  year={2019}
}

@article{ovadia2019can,
  title={Can you trust your model's uncertainty? evaluating predictive uncertainty under dataset shift},
  author={Ovadia, Yaniv and Fertig, Emily and Ren, Jie and Nado, Zachary and Sculley, David and Nowozin, Sebastian and Dillon, Joshua and Lakshminarayanan, Balaji and Snoek, Jasper},
  journal={Advances in Neural Information Processing Systems (NeurIPS)},
  volume={32},
  year={2019}
}

@article{lakshminarayanan2017simple,
  title={Simple and scalable predictive uncertainty estimation using deep ensembles},
  author={Lakshminarayanan, Balaji and Pritzel, Alexander and Blundell, Charles},
  journal={Advances in Neural Information Processing Systems},
  volume={30},
  year={2017}
}

@inproceedings{gal2016dropout,
  author    = {Gal, Yarin and Ghahramani, Zoubin},
  title     = {Dropout as a {B}ayesian Approximation: Representing Model Uncertainty in Deep Learning},
  booktitle = {Proceedings of the 33rd International Conference on Machine Learning},
  pages     = {1050--1059},
  year      = {2016}
}

@article{lewbel2007estimation,
  title={Estimation of average treatment effects with misclassification},
  author={Lewbel, Arthur},
  journal={Econometrica},
  volume={75},
  number={2},
  pages={537--551},
  year={2007},
  publisher={Wiley Online Library}
}

@article{mahajan2006identification,
  title={Identification and estimation of regression models with misclassification},
  author={Mahajan, Aprajit},
  journal={Econometrica},
  volume={74},
  number={3},
  pages={631--665},
  year={2006},
  publisher={Wiley Online Library}
}

@article{robinson1988root,
  author  = {Robinson, Peter M.},
  title   = {Root-{N}-consistent semiparametric regression},
  journal = {Econometrica},
  volume  = {56},
  number  = {4},
  pages   = {931--954},
  year    = {1988}
}

@article{chernozhukov2018double,
  author  = {Chernozhukov, Victor and Chetverikov, Denis and Demirer, Mert and Duflo, Esther and Hansen, Christian and Newey, Whitney and Robins, James},
  title   = {Double/debiased machine learning for treatment and structural parameters},
  journal = {The Econometrics Journal},
  volume  = {21},
  number  = {1},
  pages   = {C1--C68},
  year    = {2018}
}

@article{kallus2022assessing,
  title={Assessing algorithmic fairness with unobserved protected class using data combination},
  author={Kallus, Nathan and Mao, Xiaojie and Zhou, Angela},
  journal={Management Science},
  volume={68},
  number={3},
  pages={1959--1981},
  year={2022},
  publisher={INFORMS}
}

@article{breiman2001random,
  author  = {Breiman, Leo},
  title   = {Random forests},
  journal = {Machine Learning},
  volume  = {45},
  pages   = {5--32},
  year    = {2001}
}

@article{wright2017ranger,
  author  = {Wright, Marvin N. and Ziegler, Andreas},
  title   = {ranger: A fast implementation of random forests for high dimensional data in {C}++ and {R}},
  journal = {Journal of Statistical Software},
  volume  = {77},
  number  = {1},
  pages   = {1--17},
  year    = {2017}
}

@inproceedings{kohavi1996scaling,
  author    = {Kohavi, Ron},
  title     = {Scaling up the accuracy of naive-{B}ayes classifiers: a decision-tree hybrid},
  booktitle = {ACM SIGKDD International Conference on Knowledge Discovery and Data Mining},
  pages     = {202--207},
  year      = {1996}
}

@article{wang2023freematch,
  author  = {Wang, Yidong and Chen, Hao and Heng, Qiang and Hou, Wenxin and Fan, Yue and Wu, Zhen and Wang, Jindong and Savvides, Marios and Shinozaki, Takahiro and Raj, Bhiksha and Schiele, Bernt and Xie, Xing},
  title   = {{FreeMatch}: Self-adaptive thresholding for semi-supervised learning},
  journal = {International Conference on Learning Representations (ICLR)},
  year    = {2023}
}

@article{wang2022semi,
  author  = {Wang, Xudong and Wu, Zhirong and Lian, Long and Yu, Stella X.},
  title   = {Debiased learning from naturally imbalanced pseudo-labels},
  journal = {IEEE Conference on Computer Vision and Pattern Recognition (CVPR)},
  year    = {2022}
}

@inproceedings{rizve2021defense,
  author    = {Rizve, Mamshad Nayeem and Duarte, Kevin and Rawat, Yogesh S. and Shah, Mubarak},
  title     = {In Defense of Pseudo-Labeling: An Uncertainty-Aware Pseudo-label Selection Framework for Semi-Supervised Learning},
  booktitle = {International Conference on Learning Representations (ICLR)},
  year      = {2021}
}

@inproceedings{li2021comatch,
  title={{CoMatch}: Semi-supervised learning with contrastive graph regularization},
  author={Li, Junnan and Xiong, Caiming and Hoi, Steven CH},
  booktitle={Proceedings of the IEEE/CVF international conference on computer vision},
  pages={9475--9484},
  year={2021}
}

@article{vaneeden2020semisupervised,
  title={A survey on semi-supervised learning},
  author={Van Engelen, Jesper E and Hoos, Holger H},
  journal={Machine learning},
  volume={109},
  number={2},
  pages={373--440},
  year={2020},
  publisher={Springer}
}
\end{document}